\documentclass[11pt,a4paper]{scrartcl}

\usepackage[utf8]{inputenc}
\usepackage[T1]{fontenc}
\usepackage{lmodern}

\usepackage{amsmath,amssymb,amsfonts,amsthm}
\usepackage{physics}
\usepackage{hyperref}

\usepackage{mathtools}

\DeclarePairedDelimiter{\ceil}{\lceil}{\rceil}

\usepackage{authblk}

\usepackage{color}
\usepackage{dsfont}
\usepackage{bbm}

\usepackage{enumitem}

\usepackage{graphicx}

\usepackage{cleveref}

\usepackage{aliases}

{
  \theoremstyle{definition}
  \newtheorem{condition}{Condition}
}

\crefname{condition}{condition}{conditions}

\usepackage[backend=biber]{biblatex}
\bibliography{biblio}

\begin{document}

\title{Divide and conquer method for proving gaps of frustration free Hamiltonians}

\author[1]{Michael J. Kastoryano}
\affil[1]{NBIA, Niels Bohr Institute, University of Copenhagen, Blegdamsvej 17, 2100 Copenhagen, Denmark}
\author[2,1]{Angelo Lucia\thanks{angelo@math.ku.dk}}
\affil[2]{QMATH, Department of Mathematical Sciences, University of Copenhagen, Universitetsparken 5, 2100 Copenhagen, Denmark}

\maketitle

\begin{abstract}
  Providing system-size independent lower bounds on the spectral gap of local Hamiltonian is in general
  a hard problem. For the case of finite-range, frustration free Hamiltonians on a spin lattice of
  arbitrary dimension, we show that a property of the ground state space is sufficient to obtain such a bound.
  We furthermore show that such a condition is necessary and equivalent to a constant spectral gap.
  Thanks to this equivalence, we can prove that for gapless models
  in any dimension, the spectral gap on regions of diameter $n$ is at most
  $o\qty(\frac{\log(n)^{2+\epsilon}}{n})$ for any positive $\epsilon$.
\end{abstract}



\section{Introduction}

Many-body quantum systems are often described by local Hamiltonians on a lattice, in which every
site interacts only with few other sites around it, and the range of the interactions is given in
terms of the metric of the lattice.
One of the most important properties of these Hamiltonians is the so-called \emph{spectral gap}:
the difference between the two lowest energy levels of the operator. The low-temperature
behavior of the model (and in particular of its ground states) relies on whether the spectral gap
is lower bounded by a constant which is independent on the number of particles (a situation usually
referred to as \emph{gapped}), or on the contrary the spectral gap tends to zero as we take the
number of particle to infinity (the \emph{gapless}\footnote{We are using the terminology as it is
  frequently used in the quantum information community. In other contexts, one could only be
  interested in the thermodynamic limit, and the situation we have denoted as gapless does not
  necessarily imply that there is a continuous spectrum above the groundstate energy in such limit.}
case).

Quantum phase transitions are described by points in the phase diagram were the spectral gap
vanishes\cite{2012arXiv1203.4565S,Bachmann2011}, and therefore understanding the behavior of the
spectral gap is required in order to classify different phases of matter. A constant spectral gap
implies exponential decay of correlations in the groundstate\cite{Hastings_2006,Nachtergaele_2006},
and it is conjectured (and proven in 1D) that entanglement entropy will obey an area law
\cite{Hastings_2007}.  Moreover, the computational complexity of preparing the groundstate via an
adiabatic preparation scheme\cite{Farhi2000} is given by the inverse of the spectral gap, implying
that groundstates of gapped models can be prepared efficiently.  It is also believed that it is
possible to give synthetic descriptions of such groundstates in terms of Projected Entangled Paris
States (PEPS)\cite{Schuch2010}, and to prepare them with a quantum computer\cite{Schwarz2013}.

Because of the importance of the spectral gap, there is a large history of powerful results in
mathematical physics regarding whether some systems are gapped or not, such as the
Lieb-Schultz-Mattis theorem \cite{Lieb_1961} and its higher dimensional generalization
\cite{Nachtergaele_2007,Hastings_2004}, the so-called \emph{``martingale method''} for spin chains
\cite{Nachtergaele1996}, the local gap thresholds by Knabe \cite{Knabe1988} and by Gosset and
Mozgunov \cite{Gosset2016}. Cubitt, Perez-Garcia and Wolf have shown \cite{Undecidability} that the
general problem of determining, given a finite description of the local interactions, whether a 2D
local Hamiltonian is gapped or not is undecidable. Nonetheless, this result does not imply that it
is not possible to study the spectral gap of some specific models, and the problem can be decidable
if we restrict to specific sub-classes of interactions.

While these results have constituted tremendous progress, there is still a lack of practical tools for
studying the gap for large classes of lattice systems, especially in dimensions greater than one.  In
this paper we consider frustration-free, finite range local Hamiltonian on spin lattices, and we
present a technique for proving a lower bound on the spectral gap. Compared to the other methods for
bounding the spectral gap that are available in the literature, the one we propose uses a recursive
strategy that is more naturally targeted to spin models in dimension higher than 1, and which we
hope might allow to generalize some of the results that at the moment have only been proved in 1D.

The approach we present is based on a property of the groundstate space reminiscent of the
``martingale method''. A description of the groundstate space might not be available in all cases,
but it is easily obtained for tensor network models such as PEPS\cite{Schuch2010}. We are able to
prove that this condition is also necessary in gapped systems, obtaining an equivalence with the
spectral gap. More specifically, we will define two versions of the martingale condition, a strong
and a weak one, and we will show that the spectral gap implies the strong one. The strong martingale
condition implies the weak one, hence completing the loop of equivalences. This ``self-improving''
loop will allow us to give an upper bound on the rate at which the spectral gap vanishes in gapless
systems, as any rate slower than that allows us to prove a constant spectral gap.

In order to prove the equivalence between the strong martingale condition and the spectral gap, we will
use a tool known as the Detectability Lemma\cite{Aharonov2011,Aharonov2009}.
We will also show that if the Detectability Lemma operator contracts the energy by a constant
factor, then the system is gapped. This condition is reminiscent of the ``converse
Detectability Lemma''\cite{Anshu2016,Gao2015},  but we do not know whether these two conditions
are equivalent.

Proving gaps of Hermitian operators has a long history in the setting of (thermal)  stochastic evolution of classical spin systems. In this
setting, there are numerous tools for bounding the spectral gap of the stochastic generator
(which in turn allows to bound the mixing time of the process) both for classical
\cite{Martinelli1994a,Martinelli1994b,Martinelli1994c,Martinelli1999,Cesi2001,DaiPra2002} and for
quantum commuting Hamiltonians\cite{Kastoryano2016}.

In the classical setting, the theorems establish an intimate link between the mixing time of a
stochastic semigroup (the Glauber dynamics) and the correlation properties in the thermal state at a
specified temperature: for sufficiently regular lattices and boundary conditions, correlations
between two observables are exponentially decaying (as a function of the distance between their
supports) \textit{if and only if} the Glauber dynamics at the same temperature mixes rapidly (in a
time $O(\log(N))$, where $N$ is the  volume of the system). All of the proofs of the classical
results in some way or another rely on showing that exponential decay of correlations implies a
Log-Sobolev inequality of the semi-group, and in the other direction, that the log-Sobolev
inequality implies a spectral gap inequality, which in turn implies exponential decay of
correlation. We will take inspiration from a weaker form of the classical theorem that shows the
equivalence between spectral gap of the semigroup and exponential decay of correlation.

The paper is organized as follows. In \cref{sec:main-results}, we will describe the main assumption
on the groundstate space that implies the spectral gap, and then we will state the main results.  In
\cref{sec:dl} we will recall some useful tools, namely the detectability lemma and its converse. In
\cref{sec:cond-c-gap}, we will finally prove the main theorem, together with the local gap threshold.

\section{Main results}\label{sec:main-results}
\subsection{Setup and notation}

Let us start by fixing the notation and recalling some common terminology in quantum spin systems.
We will consider a $D$-dimensional lattice $\Gamma$ (the standard example being $\Gamma=\Z^D$, but
the same results will hold for any graph which can be isometrically embedded in $\R^D$). At each
site $x\in \inflat$ we associate a finite-dimensional Hilbert space $\hs_x$, and for simplicity we
will assume that they all have the same dimension $d$. For every finite subset
$\lat \subset \inflat$, the associated Hilbert space $\hs_\lat$ is given by
$\otimes_{x\in \lat} \hs_x$, and the corresponding algebra of observables is
$\alg_\lat = \bound(\hs_\lat)$. If $\lat \subset \lat^\prime$ we will identify $\alg_\lat$ as the
subalgebra $\alg_\lat \otimes \I_{\lat^\prime \setminus \lat} \subset \alg_{\lat^\prime}$. If $P$
is an orthogonal projector, we will denote by $P_\perp$ the complementary projection $1-P$.

A local Hamiltonian is  a map associating each finite $\lat \subset \inflat$ to a Hermitian
operator $H_\lat$, given by
\[ H_\lat = \sum_{X \subset \lat} h(X) ,\] where $h(X) \in \alg_X$ is Hermitian.
We will denote the orthogonal projector on the groundstate space of $H_\lat$ (i.e. the
eigenprojector corresponding to the smallest eigenvalue of $H_\lat$) as $P_\lat$.
We will make the following assumptions on the interactions $h(X)$:
\begin{description}
  \item[(Finite range)] there exist a positive $r$ such that $h(X)$ is zero whenever the diameter of
    $X$ is larger than $r$. The quantity $r$ will be denoted the \emph{range} of $h$;
  \item[(Frustration freeness)] for every $X$, $h(X) P_\lat = E_0(X)P_\lat$, where $E_0(X)$ is the lowest eigenvalue of $h(X)$.
\end{description}
Note that frustration freeness implies that $P_\lat P_{\lat^\prime} = P_{\lat^\prime} $ whenever
$\lat \subset \lat^\prime$. By applying a global energy shift, we can replace $h(X)$ with $h(X) -
E_0(X)$, and we will assume that $E_0(X)=0$ for every $X$, so that $H_\lat \ge 0$.

\begin{definition}[Spectral gap]
For every $\lat$, we will denote by $\lambda_\lat$ the difference between the two lowest distinct
eigenvalues of $H_\lat$ (which, since we have assumed that 0 is the lowest eigenvalue, is the same
as the smallest non-zero eigenvalue of $H_\lat$). This quantity will be called the \emph{spectral gap} of $H_\lat$, and it
can be expressed as follows:

\begin{equation}\label{eq:variational-sg}
  \lambda_\lat  = \inf_{\ket{\phi}} \frac{\expval{H_\lat}{\phi}}{\expval{P_\lat^\perp}{\phi}}.
\end{equation}
\end{definition}
We will interpret this as a ratio of two quadratic functionals on $\hs_\lat$:
\begin{align}
  \label{eq:variance}
  \vari_\lat(\phi) &= \braket{\phi} - \expval{P_\lat}{\phi} = \expval{P_\lat^\perp}{\phi};\\
  \label{eq:dirichelet}
  \diri_\lat(\phi) &= \expval{H_\lat}{\phi}.
\end{align}
We will use the symbol $\vari_\lat(\phi)$ since the functional can be thought as a type of variance:
it equals $\norm{\ket{\phi} - P_\lat \ket{\phi}}^2$, it is always positive and vanishes only on states in $P_\lat$.
We can then rewrite \cref{eq:variational-sg} as the following optimization problem: $\lambda_\lat$ is the largest constant such that $\lambda_\lat \vari_\lat(\phi) \le \diri_\lat(\phi)$.

In order to simply the proofs, we will also make the following extra assumption on the interactions
$h(X)$:
\begin{description}
\item[(Local projections)] Every $h(X)$ is an orthogonal projection.
\end{description}

\begin{remark}
  The assumption that every $h(X)$ is an orthogonal projection is not a fundamental restriction.
  Let us denote by $E_1(X)$ (resp. $E_{\max}(X)$) the second-smallest
  eigenvalue (resp. the largest eigenvalue) of $h(X)$, and remember that we have assumed that the
  lowest eigenvalue of each $h(X)$ is zero. If we then assume the two following conditions
  \begin{description}
  \item[(Local gap)] $e = \inf_X E_1(X) >0$;
  \item[(Local boundness)] $E=\sup_X E_{\max}(X) < \infty$;
  \end{description}
  then we can see that for every finite $\lat\subset\inflat$:
  \[ e \sum_{X\subset \lat} P_X^\perp \le H_\lat \le E \sum_{X \subset \lat} P_X^\perp, \]
  where we have denoted by $P_X$ the projector on the groundstate of $h(X)$. Therefore $H_\lat$ will
  have a non-vanishing spectral gap if and only the spectral gap of the Hamiltonian composed of
  projectors $\sum_X P_X^\perp$ is not vanishing.
  This shows that, as far as we are interested in the behavior of the spectral gap,
  requiring local gap and local boundness is equivalent to requiring that the interactions $h(X)$
  are projectors.
\end{remark}

Given a local Hamiltonian $H$ which is finite range and frustration free
it is easy to see that interactions can be partitioned into  $g$ groups,
referred to as ``layers'', in such a way that every layer consists of non-overlapping (and therefore
commuting) terms. For a fixed $\lat \subset \inflat$, let us index the layers from 1 to $g$, and
denote  $L_i$ the orthogonal projector on the common groundstate space of the interactions
belonging to group $i$. Since they are commuting, $L_i$ can also be seen as the product of the
groundstate space projectors of each interaction term.  For any given ordering of $\{1,\dots, g\}$, we can then define the product
$L = \prod_{i=1}^g L_i$ (different orders of the product will in general give rise to different
operators). Any operator constructed in this fashion is called an \emph{approximate ground state
  projector}.

\begin{figure}[h]
\centering
  \includegraphics[scale=0.45]{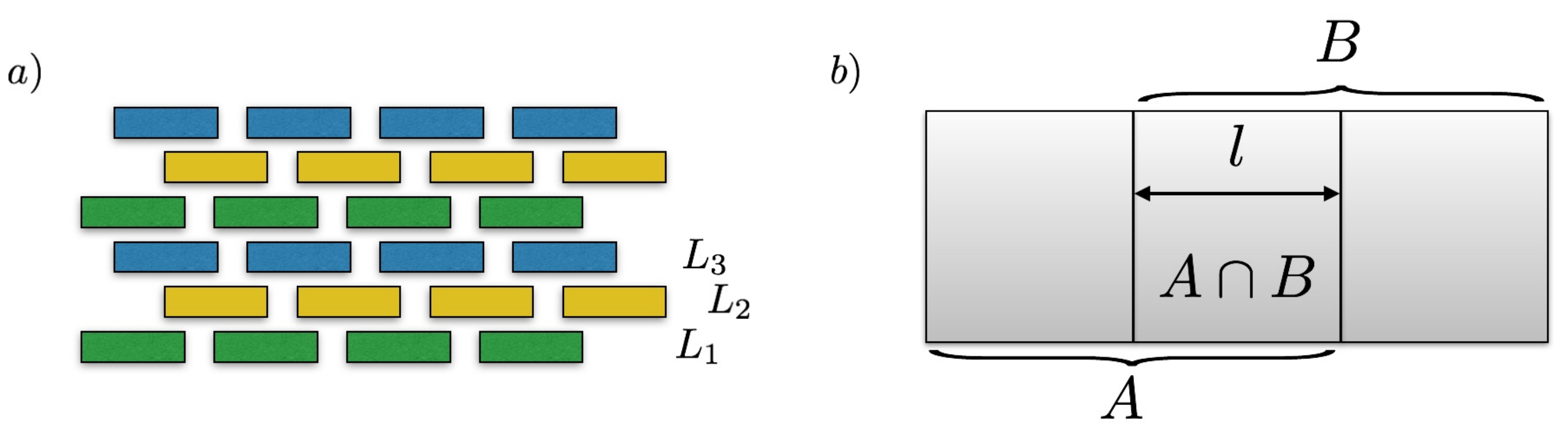}
    \caption{a) Depiction of two $g=3$ layers covering of the one dimensional lattice by local orthogonal projectors $L_g$. b) Decomposition of the lattice region $A\cup B$ in the definition of the martingale condition.   }
    \label{fig:1}
\end{figure}

\subsection{Statement of the results}\label{sec:results}

We will now state the main assumptions needed in the proof of the
spectral gap theorem.  In order to do so, we will need to introduce
some notation for the overlap between groundstate spaces of different
regions.

\begin{definition}\label{def:delta}
Let $A, B$ be finite subsets of $\inflat$. Let $P_{A\cup B}, P_A,$ and $P_B$ be
respectively the orthogonal projectors on the ground state space of $H_{A\cup B}, H_A$
and $H_B$. Then we define
\begin{equation}
  \delta(A,B) = \norm{(P_A - P_{A\cup B}) (P_B - P_{A\cup B})}.
\end{equation}
\end{definition}
\begin{remark}
Because of  frustration freedom, we have $P_A P_{A\cup B} = P_{A\cup B} P_A = P_{A\cup B}$
and the same holds for $P_B$. In turn this imply that
\[ (P_A - P_{A\cup B})(P_B - P_{A\cup B}) = P_A P_B - P_{A\cup B} ,\] so that $\delta(A,B)$ can be
both seen as a measure of the overlap between $(P_A-P_{A\cup B})$ and $(P_B - P_{A\cup B})$ (the
cosine of the first principal angle between the two subspaces), as well as a measure of how much
$P_{A\cup B}$ can be approximated by $P_AP_B$.
\end{remark}

The intuition behind Def. \ref{def:delta} is that in a gapped system, if $l$ is
the diameter of the largest ball contained in $A\cap B$, then $\delta(A,B)$ should be a fast
decaying function of $l$. In this setting we will refer to the ``size'' of the overlap of $A$ and
$B$ as $l$ (see Fig 1b). One might also hope that $\delta(A,B)$ only depends on $l$ and not on the size
of $A \Delta B = (A \cup B) \setminus (A\cap B)$. This is captured by the following assumption:

\begin{condition}\label{cond:a}
  There exists a positive function $\delta(l)$ with exponential decay in $l$, i.e.
  $\delta(l) \le c\alpha^l$ for some $0<\alpha<1$ and $c>0$, such that
  for every connected $A$ and $B$, such that $A\cap B$  has size $l$, the following bound holds:
  \begin{equation}
    \label{eq:martingale-condition}
    \delta(A,B) \le \delta(l).
  \end{equation}
\end{condition}

We will now present some weaker versions of \cref{cond:a}. As we will show later, they will all turn
out to be equivalent, but it might be hard to verify the stronger versions in some concrete
examples. The first relaxation we have is to require a slower decay of the function $\delta(l)$.

\begin{condition}\label{cond:b}
  There exists a positive function $\delta(l)$ with polynomial decay in $l$, i.e.
  $\delta(l) \le cl^{-\alpha}$ for some $\alpha>0$ and $c>0$, such that
  for every connected $A$ and $B$, such that $A\cap B$  has size $l$,
  \cref{eq:martingale-condition} holds.
\end{condition}

Clearly, \cref{cond:a} implies \cref{cond:b}. As formulated, conditions \ref{cond:a} and
\ref{cond:b} and B require \cref{eq:martingale-condition} to be satisfied homogenously for all
regions $A$ and $B$ of arbitrary size. However, in order to prove a bulk spectral gap, such a strong
homogeneity assumption can be relaxed. We can allow for the size of $A\cap B$, of $A$ and of $B$ to
be taken into account; intuitively, we would like to have less stringent requirement if $A\cap B$ is
very small compared to $A$ and $B$.  In particular, we will define classes $\mcl F_k$ of sets, which
have the property that they can be decomposed as overlapping unions of sets in $\mcl F_{k-1}$, with
a sufficiently large overlap. Then we will only require \cref{eq:martingale-condition} to hold for
this specific decomposition, and moreover we will allow the bound $\delta(l)$ to depend on $k$.

The construction of the sets $\mcl F_k$ we present is a generalization of the one originally proposed by Cesi
\cite{Cesi2001} and used in the context of open quantum systems by one of the authors \cite{Kastoryano2016}.
\begin{definition}\label{def:fk}
For each $k \in \N$, let $l_k = (3/2)^{k/D}$ and denote
\[ R(k) = [0,l_{k+1}] \times \dots \times [0,l_{k+D}] \subset \R^D. \]
Let $\mcl F_k$ be the collection of $\Lambda \subset \inflat$ which are contained in $R(k)$ up to translations and permutation of the coordinates.
\end{definition}

\begin{figure}[h]
\centering
  \includegraphics[scale=0.45]{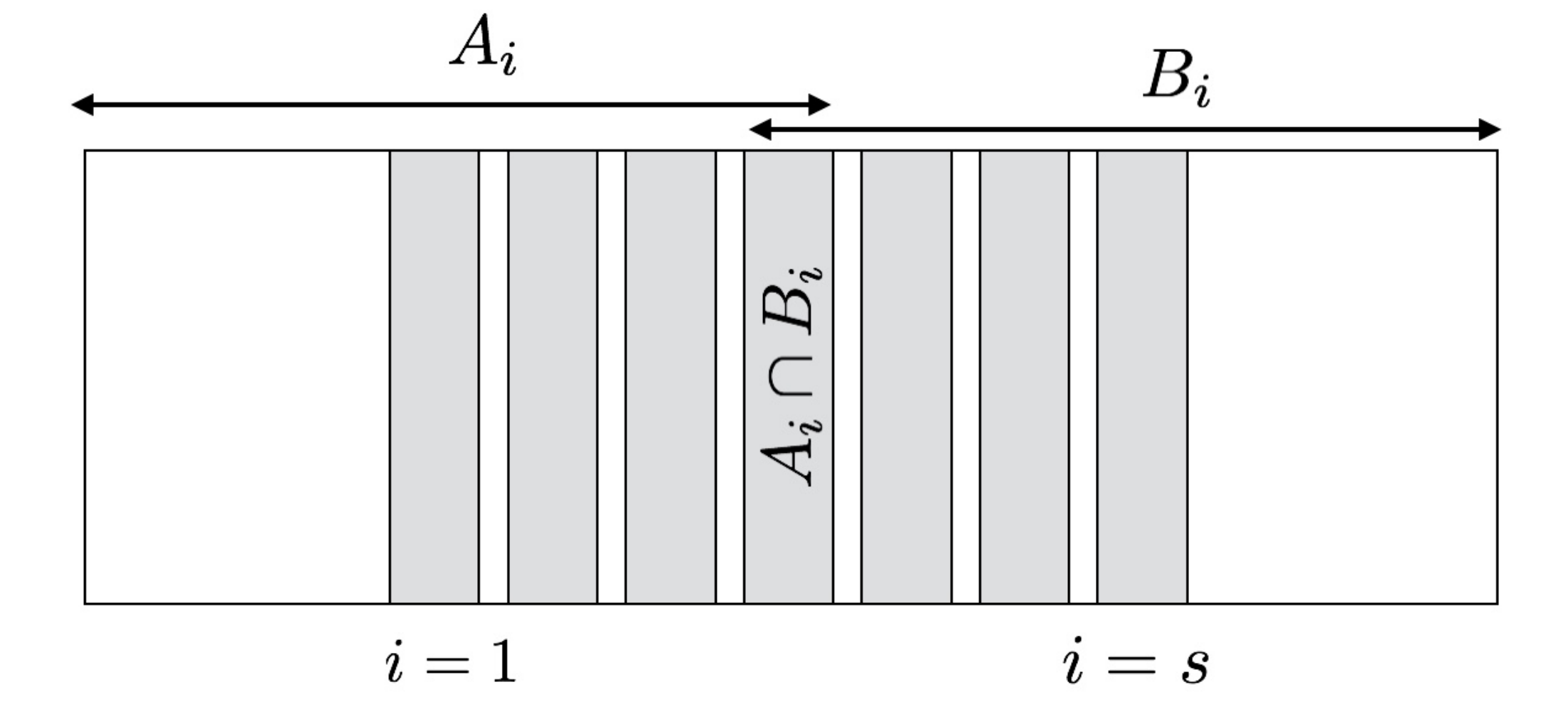}
    \caption{Depiction of the decomposition of the region $\Lambda=A_i\cup B_i$, with $i=1,...,s$, where the intersections of sets $A_i\cap B_i$ are all non-overlapping.}
    \label{fig:2}
\end{figure}

We now show that  sets in $\mcl F_k$ can be decomposed ``nicely'' in
terms of sets in $\mcl F_{k-1}$.
\begin{proposition}\label{prop:geometrical-construction}
  For each $\Lambda \in \mcl F_k \setminus \mcl F_{k-1}$ and each positive integer $s \le \frac{1}{8} l_k$, there exist
  $s$ distinct pairs of non-empty sets $(A_i,B_i)_{i=1}^s$ such that
\begin{enumerate}
  \item $\Lambda = A_i \cup B_i$ and $A_i, B_i \in \mcl F_{k-1} \quad \forall i=1,\dots,s$;
  \item $\dist(\Lambda \setminus A_i, \Lambda \setminus B_i) \ge \frac{l_k}{8s} -2$;
  \item $A_i \cap B_i \cap A_j \cap B_j = \emptyset \quad \forall i \neq j$.
  \end{enumerate}
  We will call a set of $s$ distinct pairs $(A_i, B_i)_{i=1}^s$ of non-empty sets satisfying the above properties an \emph{$s$-decomposition of $\Lambda$}.
\end{proposition}

The proof of this proposition -- a minor variation over the one presented by Cesi\cite{Cesi2001} -- is contained in \cref{appendix:geometrical-construction}.
With this definition of $\mcl F_k$ at hand, we can now present the weakest version of \cref{cond:a}.
\begin{condition}\label{cond:c}
  There exists an increasing sequence of positive integers $s_k$, with $\sum_{k}\frac{1}{s_k} <
  \infty$, such that
  \begin{equation}\label{eq:cond-c}
    \sum_{k=1}^\infty \delta_k := \sum_{k=1}^\infty \sup_{\lat \in \mcl F_k \setminus
      F_{k-1}}\sup_{A_i, B_i} \delta(A_i,B_i) < \infty,
  \end{equation}
  where the second supremum is taken over all $s_k$-decompositions $\lat = A_i \cup B_i$ given by \cref{prop:geometrical-construction}.
\end{condition}

It is not immediately clear from the definition that \cref{cond:c} is implied by \cref{cond:b}, so
we show this in the next proposition.
\begin{proposition}
  \Cref{cond:b} implies \cref{cond:c} with any $s_k$ such that $\sum_{k}\frac{s_k}{l_k} < \infty$.
\end{proposition}
\begin{proof}
  Let $\delta(l)$ be as in \cref{cond:a}. Since for every
  $\lat_k\in \mcl F_k \setminus \mcl F_{k-1}$ and for every $s_k$-decomposition
  $\lat_k=A_i \cup B_i$ of $\lat_k$, the overlap $A_i \cap B_i$ has size at least
  $\frac{l_k} {8s_k}-2$, then

  \[ \delta_k \le \delta\qty(\frac{l_k}{8s_k}  - 2) .\]
  Since $\delta(l)$ decays as $l^{-\alpha}$ for some positive $\alpha$, $\delta_k$ is
  summable if $\sum s_k/{l_k}$ is summable.
\end{proof}

\begin{remark}\label{remark:cond-a-cond-c}
  If we consider \cref{cond:c} with $s_k$ growing faster than $\frac{l_k}{k}$, then the previous
  proposition does not apply -- note that in any case $s_k$ has to be smaller than $\frac{1}{8} l_k$
  for the construction of \cref{prop:geometrical-construction} to be possible. In practice we do not
  need to consider such situations. In Cesi\cite{Cesi2001}, $s_k$ was chosen to be of order
  $l_k^{1/3}$. As we will see later, we will be interested in choosing $s_k$ with slower rates than
  that (while still having $\sum 1/s_k$ finite), so the condition $s_k = \order{\frac{l_k}{k}}$ will
  not be restrictive for our purposes. So from now on, we will only consider \cref{cond:c} in the
  case where  $s_k = \order{\frac{l_k}{k}}$.
\end{remark}

The main result of the paper is to show that \cref{cond:c} is sufficient to prove a
spectral gap. In turn, this will imply \cref{cond:a}, which as we have already seen in
\cref{remark:cond-a-cond-c} implies \cref{cond:c}, showing that all  three conditions are equivalent.

\begin{theorem}
  \label{thm:main-intro}
Let $H$ be a finite range, frustration free, local Hamiltonian, and let $\mcl F_k$ be as in \cref{prop:geometrical-construction}. Then the following are equivalent
\begin{enumerate}
  \item $\inf_k \inf_{\lat \in \mcl F_k} \lambda_\lat \ge \lambda > 0$ (or in other words, $H$ is gapped);
  \item $H$ satisfies \cref{cond:a} with $\delta(\ell) = \frac{1}{(1+\lambda/g^2)^{l/2}}$ for some
    constant $g$;
  \item $H$ satisfies \cref{cond:b};
  \item $H$ satisfied \cref{cond:c} with $s_k$ such that $ \sum_{k}\frac{s_k}{l_k} < \infty$.
\end{enumerate}
\end{theorem}

By proving the equivalence of these conditions, we are also able to show that in any gapless model, the spectral gap cannot close too slowly, since a slow enough (but still infinitesimal) gap will imply \cref{cond:c} and therefore a constant gap. The threshold is expressed in the following corollary.
\begin{corollary}\label{cor:threshold}
  If $H$ is gapless, then for any $\lat \subset \inflat$ of diameter $n$ it holds that
\[
    \lambda_\lat = o\qty(\frac{\log(n)^{2+\epsilon}}{n}),
  \]
  for every $\epsilon>0$.
\end{corollary}

We also provide an independent condition for lower bounding the spectral gap.
Consider again the construction of the detectability lemma, where $L=L_1\cdots L_g$ is an approximate ground state projector.
\begin{theorem}\label{THM:M2-GAP}
If there exist a constant $0<\gamma<1$ such
\be  \diri(L\phi) \leq \gamma \diri(\phi),\ee
then the spectral gap of $H$ is bounded below by $\lambda \ge \frac{1-\gamma}{4}$.
\end{theorem}
While similar in spirit to the Converse Detectability Lemma (see \cref{sec:dl}), we do not know if
these are equivalent, nor whether the hypothesis of \cref{THM:M2-GAP} is  necessary.

\begin{remark}[Comparison with the ``martingale method'']
  Nachtergaele \cite{Nachtergaele1996} presented a general method for proving the spectral gap for
  a class of spin-lattice models, which has become known as the \emph{martingale method}.  Given
  an increasing and absorbing sequence $\lat_n \to \inflat$, and a fixed parameter $l$, it requires three conditions ((C1), (C2), (C3)) to be
  satisfied uniformly along the sequence to prove a lower
  bound to the spectral gap.  Let us briefly recall what these conditions would be if we applied
  them to the setting we are considering, and compare them to \cref{cond:c}.  The first
  condition, denoted (C1) in the original paper, is automatically satisfied by finite range
  interactions, which is also the case we are considering here.  If we denote $A_n = \lat_n$ and
  $B_n = \lat_{n+1}\setminus \lat_{n-l}$ (where now $l$ is a parameter partially controlling the
  size of $A_n \cap B_n = \lat_n \setminus \lat_{n-l}$), then condition (C2) requires that
  $H_{A_n \cap B_n}$ has a spectral gap of $\gamma_l$ independently of $n$ (for every $n$ large
  enough). We do not need to require such assumption, since we are using a recursive proof.
  Condition (C3) can be restated, using our notation, as requiring that
  $\delta(A_n, B_n) \le \epsilon_l < \frac{1}{\sqrt{l+1}}$ for all $n$ large enough.

  Clearly, the big difference with \cref{cond:c} is that the requirement on $\delta$ is not of
  asymptotic decay, but only to be bounded by a specific constant.  Upon careful inspection,
  we see this is only a fair comparison in 1D. In higher dimensions, condition (C2) could be as hard
  to verify as the original problem of lower bounding the spectral gap, since the size of
  $A_n \cap B_n$ will grow with $n$. Condition (C3) is also clearly implied by \cref{cond:a}.
  Therefore, one could compare the method we propose with the martingale method as a strengthening
  of condition (C3) in exchange of a weakening of condition (C2), a trade-off which we hope makes it
  more applicable in dimensions $D>1$.
\end{remark}

\subsection{Example 1: translation invariant 1D spin chains}

To clarify the differences between \cref{cond:a,cond:b,cond:c}, let us consider the case of 1D spin
chains. We will consider a translational invariant model to further simplify the situation. Then we
can take, without loss of generality, $A=[0,n]$ and $B=[n-d,n-d+m]$, with $n$, $m$, $d$ being positive
integers such that $\min(m,n) > d$. The intersection $A\cap B = [n-d,n]$ has length $d+1$, so
that \cref{cond:a} is equivalent to the fact that the function
\begin{equation}\label{eq:example-cond-a}
  \delta(d) = \sup_{d<n,m} \delta([0,n], [n-d,n-d+m])
\end{equation}
has exponential decay in $d$. \Cref{cond:b} would relax this to a polynomial decay, but both require
a bound that is uniform in $n$ and in $m$.

We can now consider the larger interval in each $\mcl F_k$, namely $\Lambda_k = [0,(3/2)^{k+1}]$. Denoting $l_k
= (3/2)^k$, we can write $\Lambda_k$ and its $s$-decompositions as
\begin{align*}
  \Lambda_k &= \qty[0, 1] \cdot l_{k+1} \\
  A^i_k &= \qty[0, \frac{1}{2}  + \frac{i}{6s} ] \cdot l_{k+1}  \\
  B^i_k &= \qty[\frac{1}{2} + \frac{i}{6s} - \frac{1}{12s}, 1] \cdot l_{k+1},
\end{align*}
for $i=1, \dots, s$. The overlap $A^i_k \cap B^i_k$ has size $\frac{l_k}{12s}$ for every $i$. If we
fix for concreteness $s_k = l_k^{1/3}$, as in \cite{Cesi2001}, then we can define
\[ n_{i,k} = \ceil*{ \frac{1}{2}l_k + \frac{i}{6} l_k^{2/3}}, \quad m_{i,k} =
  \ceil*{\frac{1}{2}l_k - \frac{2i-1}{12} l_k^{2/3}}, \quad d_k = \ceil*{\frac{l^{2/3}_k}{12}}, \]
so that
\[
  A^i_k = [0, n_{i,k}] \quad B^i_k = [n_{i,k} + d_k, n_{i,k} + d_k + m_{i,k}].
\]
Note that $n_{i,k}$ and $m_{i,k}$ are always smaller than $24 \sqrt{3}\, d_k^{3/2}$
So we then see that in order to show that the model satisfies \cref{cond:c}, it would be sufficient
for example to verify that
\begin{equation}\label{eq:example-cond-c}
  \delta(d) = \sup_{d<n,m \le 24 \sqrt{3}\, d^{3/2}} \delta([0, n ], [n -d, n - d + m ])
\end{equation}
is decaying polynomially fast in $d$. Compared to \cref{eq:example-cond-a}, $n$ and $m$ are
restricted given a specific $d$, i.e. we only have to consider the case where they are at most a
constant times $d^{3/2}$.
It should be clear now that this restriction on the $n$ and $m$ depends on the choice of the scaling
of $s_k$. Choosing faster rates of growth for $s_k$ leads to more restrictive conditions (and thus in principle
easier to verify): the downside is that this will be reflected in the numerical bound for the
spectral gap, which will become worse (although finite).

\subsection{Example 2: PVBS models}
One notable model in dimension larger than 1 for which the original martingale method has been
successfully applied is the Product Vacua and Boundary State (PVBS) model
\cite{Bachmann_2015,Bishop2016}, a translation invariant, finite range, frustration free spin
lattice Hamiltonian, with parameters $D$ positive real numbers $(\lambda_1, \dots, \lambda_D )$.  We
refer to the original paper for the precise definition of the model.  The spectral gap of the PVBS
Hamiltonian is amenable to be analyzed using the ``1D version'' of the martingale method, applied
recursively in each of the dimensions, and it has been shown that in the infinite plane the
Hamiltonian is gapped if and only if not all $\lambda_j$ are equal to $1$.  In this section we show
that our result recovers the same finite-size limit analysis as in the original paper: for
simplicity we will only do the analysis in the case of rectangular regions, with the caveat that in that case
the finite-size gap closes if only one of the $\lambda_j$ is equal to $1$ (even if the GNS Hamiltonian is still gapped).

  In Ref. \cite[Lemma 3.3]{Bishop2016} it has been shown that
  in the case of two connected regions $A$ and $B$ such that $A\cap B$ is also connected,
  \begin{equation}\label{eq:pvbs-martingale}
    \delta(A,B)^2 = \frac{C(A\setminus B) C(B\setminus A)}{C(A)C(B)},
  \end{equation}
  where $C(X) = \sum_{x \in X} \prod_{j=1}^D \lambda_j^{2x_j}$ is the normalization constant of the
  model.  If we now consider $\lat \in \mcl F_k$ to be a rectangular region (so that every $A_i$ and
  $B_i$ appearing in the geometrical construction of \ref{prop:geometrical-construction} will also
  be rectangles), then the normalization constant $C(\lat)$ will be a product of different constants
  in each dimension independently. Assuming without loss of generality that the dimension being cut
  by \cref{prop:geometrical-construction} is the $D$-th, we see that if $\lambda_D = 1$ then
  \[ \delta(A_i,B_i) = \left( \frac{\abs{A_i\setminus B_i}}{\abs{A_i}} \frac{\abs{B_i\setminus
          A_i}}{\abs{B_i}} \right)^{1/2} = \order{1-\frac{1}{8s_k}} ,\]
  which is not infinitesimal.

  On the other hand, if $\lambda_D \neq 1$, then
  \[ \delta(A_i, B_i) = \left( \frac{\sum_{x=0}^{l_A-l} \lambda_D^{2x} }{\sum_{x=0}^{l_A} \lambda_D^{2x}}
      \frac{\sum_{x=l}^{l_B} \lambda_D^{2x} }{\sum_{x=0}^{l_B} \lambda_D^{2x}} \right)^{1/2} ,\]
  have denoted by $l_A$ (resp. $l_B, l$) the length of $A$ (resp. $B$, $A\cap B$) along dimension
  $D$. Therefore
  \[ \delta(A_i,B_i)\le
    \begin{cases}
      \lambda_D^{l+1} [(1-\lambda_D^{2(l_A+1)})(1-\lambda_D^{2(l_B+1)})]^{-1/2} & \text{if } \lambda_D < 1,\\
      \lambda_D^{-(l+1)} [(1-\lambda_D^{-2(l_A+1)})(1-\lambda_D^{-2(l_B+1)})]^{-1/2} & \text{if } \lambda_D > 1.
    \end{cases}
  \]
  If all $\lambda_j$ are distinct from one, then the PVBS model satisfies \cref{cond:a} with
  \[ \delta(l) = \frac{\lambda_*^{l}}{1-\lambda^2_*} , \quad \lambda_* = \max_i \min(\lambda_i, \lambda_i^{-1}), \]
  and therefore it is gapped by \cref{thm:cond-c-gap}. If at least one of them is equal to $1$
  then $\delta(l)$ will be lower bounded away from zero, and therefore the gap will close.
  Note that one could get a better estimate on the spectral gap by following the proof of
  \cref{thm:cond-c-gap}, and using a different $\delta(l)$ in each of the dimensions, instead that just taking
  the worst case as we did here.

\section{Detectability lemma and spectral gap}\label{sec:dl}

\subsection{The detectability lemma and its converse}

The Detectability Lemma \cite{Anshu2016,Aharonov2011,Aharonov2009} originated in the
context of the quantum PCP conjecture\cite{Aharonov2013}.  It has since then become a useful tool in
many-body problems.  A converse result is known as the Converse Detectability Lemma
\cite{Gao2015,Anshu2016}, and will also be used later. At the same time as we recall them, we will
reformulate them in terms of inequalities between some quadratic functionals.

In analogy to \cref{eq:variance}, given $L=\prod_{i=1}^g L_i$ we define the following quadratic functional on $\hs_\lat$
\begin{equation}
  \label{eq:dl-def}
  DL(\phi) = \braket{\phi} - \expval{L^*L}{\phi}.
\end{equation}
Before stating the Detectability Lemma and its converse, let us make some preliminary observations regarding $L$ and $DL(\phi)$.
\begin{remark}
  For any $L$ given above, denote $P$ the projector on the groundstate space of $H$. Then
  \begin{enumerate}[label={(\arabic*)}]
  \item $L P = P L = P$, and in particular $\comm{L}{P} = 0$;
  \item $\norm{L} \le 1$;
  \end{enumerate}
\end{remark}
\begin{proof}
  (1) follows from the definition of $L$ and frustration
  freedom. Since $L$ is a product of projectors its norm is bounded
  by 1, so also (2) is trivial.
  \end{proof}

\begin{proposition}\label{remark:dl-vari}
  For every $\phi \in \hs_\Lambda$ it holds that
  \begin{equation}\label{eq:dl-vari}
    DL(\phi) \le \vari(\phi) \le \frac1{1-\norm{LP^\perp}^2} DL(\phi),
  \end{equation}
  and $1/(1-\norm{LP^\perp}^2)$ is the smaller constant that makes the upper bound hold true.
\end{proposition}
\begin{proof}
  Let us start by observing that
  \[ \vari(L \phi) = \expval{L^*P^\perp L}{\phi} = \expval{L^*L}{\phi} - \expval{P}{\phi} =
    \vari{\phi} - DL(\phi). \] On the one hand, since $\vari$ is a positive quadratic functional, we
  have that $\vari(L\phi) \ge 0$ and therefore $\vari(\phi) \ge DL(\phi)$.  On the other hand we
  have the following bound
  \[ \vari(\phi) - DL(\phi) = \expval{L^*P^\perp L}{\phi} = \expval{P^\perp L^*L P^\perp}{\phi} \le \norm{LP^\perp}^2 \expval{P^\perp}{\phi};  \]
  from which the upper bound in \cref{eq:dl-vari} follows by rearranging the terms. Optimality follows by choosing a $\phi$ such that $\norm{LP^\perp \phi} = \norm{LP^\perp} \norm{P^\perp\phi}$.
\end{proof}
As can be seen from \cref{eq:dl-vari}, if $\norm{LP^\perp}$ is smaller than 1, then $DL$ is up to constants equivalent to $\vari$.
The Detectability lemma and its converse then relate $DL$ to $\diri$, thus allowing to connect $\norm{LP^\perp}$ to the spectral gap, via
\cref{eq:variational-sg}.

\begin{lemma}[Detectability Lemma]\label{lemma:dl}
  With the notation above, it holds that
  \begin{equation}\label{eq:dl}
    \diri(L\phi) \le g^2 DL(\phi).
  \end{equation}
\end{lemma}
The proof of this statement can be found in Ref. \cite[Lemma 2]{Anshu2016}. A simple corollary follows:
\begin{corollary}\label{cor:dl}
  If $\lambda$ is the spectral gap of $H$, then
  \begin{equation}\label{eq:dl-cor}
    \norm{L P^\perp}^2 \le \frac{1}{1+\lambda/g^2}.
  \end{equation}
  In particular, for finite systems $\norm{LP^\perp}<1$.
\end{corollary}
\begin{proof}
  If $\lambda$ is the spectral gap of $H$, then $\lambda \vari(\phi) \le \diri(\phi)$. In particular,
  $\lambda \vari(L\phi) \le \diri(L\phi) \le g^2 DL(\phi)$. But in \cref{remark:dl-vari} we have seen that $\vari(L\phi) = \vari(\phi) - DL(\phi)$,
  and therefore $\vari(\phi) \le (1+\frac{g^2}{\lambda})DL(\phi)$. The result follows from optimality of the constant in \cref{eq:dl-vari}.
\end{proof}

\begin{lemma}[Converse DL]\label{lemma:converse-dl}
  With the same notation as above,
  \begin{equation}\label{eq:converse-dl}
    DL(\phi) \le 4 \diri(\phi).
  \end{equation}
\end{lemma}
The proof of this statement can be found in Ref. \cite[Corollary 1]{Gao2015}.
Again, from this functional formulation we can derive the usual statement of the Converse Detectability lemma
\begin{corollary}\label{cor:converse-dl}
  If $\lambda$ is the spectral gap of $H$, then
  \begin{equation}
    \lambda \ge \frac{1-\norm{LP^\perp}^2}4.
  \end{equation}
\end{corollary}
\begin{proof}
  It follows from \cref{remark:dl-vari}.
\end{proof}

We are now ready to prove \cref{THM:M2-GAP}.
\begin{proof}
  From \cref{cor:dl}, we have that $\norm{L P^\perp} < 1$, and then \cref{remark:dl-vari} implies that $\lim_{n\to \infty} L^n = P$.
  Therefore
  \begin{eqnarray} \lim_{m\rightarrow\infty}\sum_{n=0}^m DL(L^{n}\phi) &=&  \lim_{m\rightarrow\infty}\sum_{n=0}^{m} \expval{(L^n)^*L^n}{\phi} - \expval{(L^{n+1})^*L^{n+1}}{\phi} \\
   &=&  \lim_{m\rightarrow\infty}\braket{\phi} -  \expval{(L^{m+1})^*L^{m+1}}{\phi}\\
    &=& \braket{\phi} - \expval{P}{\phi} = \vari(\phi).
  \end{eqnarray}
  By applying \cref{lemma:converse-dl} to each term in the summation, we obtain that:
  \[ \vari(\phi) = \sum_{n=0}^\infty DL(L^n\phi) \le 4 \sum_{n=0}^\infty \diri(L^n \phi) \le 4 \sum_{n=0}^\infty \gamma^n \diri(\phi) = \frac{4}{1-\gamma}\diri(\phi).  \]
\end{proof}

\subsection{Spectral gap implies \cref{cond:a}}
Let us start by proving the following converse relationship between spectral gap and $\delta(A,B)$.

\begin{theorem}\label{thm:gap-m2}
  Let $A, B \subset \inflat$ be finite, and let $l=\dist((A\cup B)\setminus A, (A\cup B)\setminus B)$.
  If $H_\Lambda$ for $\Lambda={A\cup B}$  is a finite range Hamiltonian with spectral gap $\lambda_\Lambda$, then
  \begin{equation}
    \label{eq:gap-m1}
    \delta(A,B) \le \frac{1}{(1+\lambda_\Lambda/g^2)^{l/2}},
  \end{equation}
  where $g$ is a constant depending only on $\inflat$ and on the range of $H$.
\end{theorem}

In order to prove this result, we will make use of the Detectability Lemma.
With the same notation as in \cref{lemma:dl}, it implies that
$ \norm{L{P_{\Lambda}^\perp}}^2 \le \frac{1}{1+\lambda_{\Lambda}/g^2}$.
By taking $q$-powers of $L$ and iterating the previous bound $q$ times we obtain
\[ \norm{L^q P_{\Lambda}^\perp} \le \frac{1}{(1+\lambda_{\Lambda}/g^2)^{q/2}} = \epsilon_\Lambda^q, \]
since $LP_{\Lambda}^\perp \subset P_{\Lambda}^\perp$, where we have denoted $\epsilon_\Lambda =
(1+\lambda_\Lambda/g^2)^{-1/2} < 1$. Therefore, if $H_{\Lambda}$ is gapped, $L^q$ will be an
exponentially good approximation of $P_{\Lambda}$, with $q$ chosen independently of $\Lambda$.
We now want to show that $L^q$ can be split as a product of two terms
$L^q = M_A M_B$ in such a  way that both $M_A$ and $M_B$ are good
approximations to $P_A$ and $P_B$, using a strategy presented in Ref. \cite{Kastoryano2016}.

\begin{lemma}\label{lemma:gap-m1}
  With the notation defined above, if $q\le l$, then there exist two operators $M_A$ and $M_B$,
  respectively acting on $A$ and on $B$, such that $L^q=M_AM_B$ and the following holds:
  \begin{equation}
    \norm{P_A - M_A} \le \epsilon_\Lambda, \quad \norm{(P_A -M_A)M_B} \le \epsilon_\Lambda^q;
  \end{equation}
  and the same holds with $A$ and $B$ interchanged.
\end{lemma}
\begin{proof}
  Let us start by defining $M_A$ and $M_B$ as follows: we will group the projectors appearing in $L^q$ in two disjoint groups, such that $M_A$ will be the product (in the same order as they appear in $L^q$) of the projectors of one group, $M_B$ the product of the rest, and $L^q=M_A M_B$.
  In order to do so, we will consider the layers $L_1, \dots, L_g$ sequentially (following the order in which are multiplied in $L$), and then we will
start again from $L_1$ up to $L_g$, until we have considered $\lfloor {gq/2} \rfloor $ different layers. Each layer will be split into two parts, where terms of one of them will end up appearing in $M_A$ and terms in the other will appear in $M_B$.
In the first layer, we will only include in $M_A$ terms which intersect $(A\cup B)\setminus B$.
From the second layer, we only included terms which intersect the support of the terms considered from the first. We keep doing this recursively, when at each layer we include terms which intersect the support of the selected terms of the previous step (one can see this as a sort of light-cone, defined by the layer structure, generated by $(A\cup B)\setminus B$, as depicted in \cref{fig:3}).
The remaining $\lceil {gq/2} \rceil$ are treated similarly, but starting instead from the end of the product, and reversing the role of $B$ and $A$.
At this point, it should be clear that by construction $L^q=M_AM_B$, since every projector appearing in $L$ has been assigned to either $M_A$ or $M_B$, and the two groups can be separated without breaking the multiplication order.
If $q \le l$, then $M_B$ will be supported in $B$, and $M_A$ will be supported on $A$.

\begin{figure}[h]
\centering
  \includegraphics[scale=0.55]{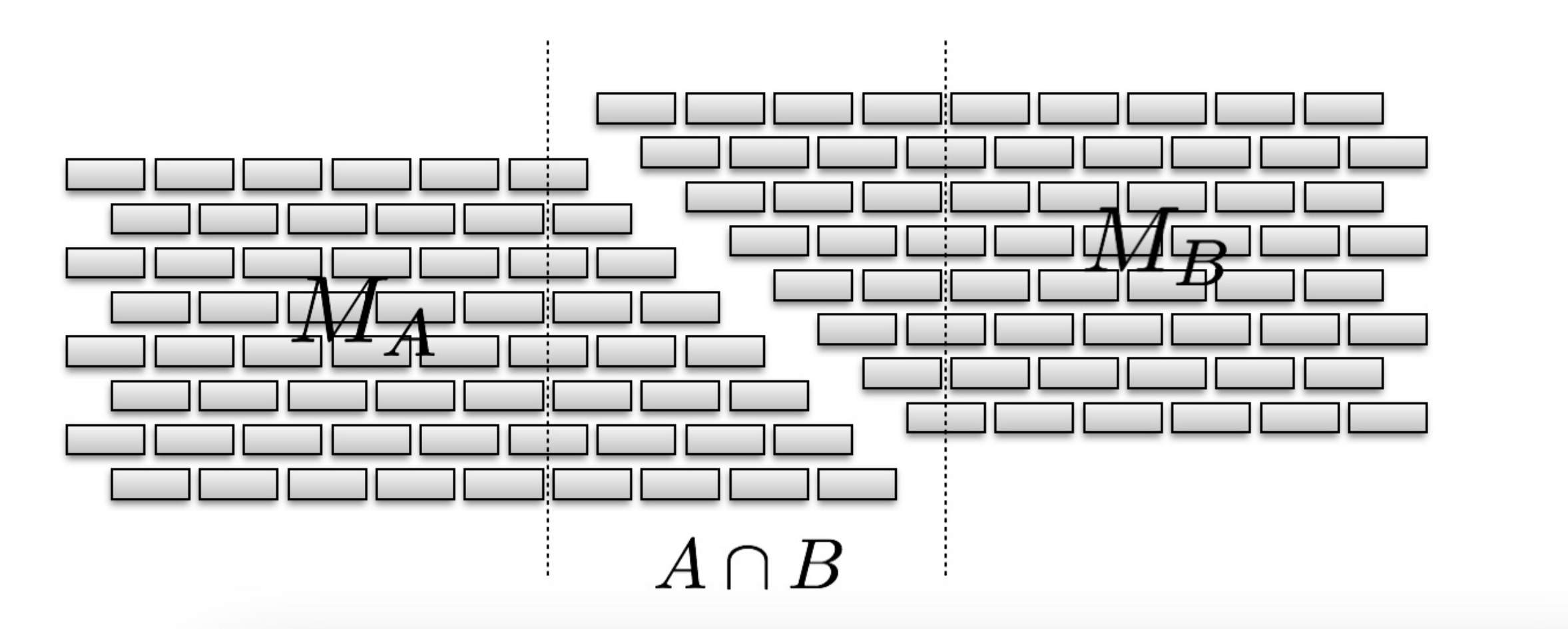}
    \caption{Depiction of the decomposition of the operator $L^q=M_A M_B$, with $g=2$ distinct layers in $L$. }
    \label{fig:3}
\end{figure}

Denote with $L_A$ and $L_B$ the approximate ground state projections of $P_A$ and $P_B$ respectively, as in \cref{lemma:dl}.
Then we have that $\norm{P_A - L_A^q} \le \epsilon_\Lambda^q$ and the same for $B$.
It should be clear that $M_A$ and $M_B$ contain strictly more projection terms than $L_A$ and $L_B$, and therefore
$P_AM_A = M_A P_A = P_A$ and $\norm{M_AP_A^\perp} \le \epsilon_\Lambda$, and the same holds for $B$.
Observe that we can write $P_AM_B = P_A M_A M_B := P_A R L_B^q$, where we
have redistributed the projectors of $M_A$ in order to ``fill'' the missing ones in $M_B$ to complete it to $L_B^q$. What is left is put into $R$,
which can be reabsorbed into $P_A$. Therefore $P_AM_B = P_AL_B^q$, and this implies that
\[ \norm{P_A(P_B-M_B)} = \norm{P_A(P_B-L_B^q)} \le \epsilon_\Lambda^q.\]
The same construction (but exchanging the roles of $A$ and $B$) can be done in order to bound
$\norm{(P_A-M_A)M_B} \le \epsilon_\Lambda^q$.
\end{proof}
With this construction, we can easily prove \cref{thm:gap-m2}.
\begin{proof}[Proof of  \cref{{thm:gap-m2}}]
We observe that
\[ P_AP_B - M_AM_B = P_A(P_B - M_B)  + (P_A - M_A)M_B . \]
We can now apply \cref{lemma:gap-m1}, and choose $q=l$ to obtain
  \[ \norm{P_A P_B - P_{A \cup B}} \le \frac{1}{(1+\lambda_\Lambda/g^2)^{l/2}} .\]
\end{proof}

In the next section, we will show that condition \cref{cond:c} implies the spectral gap.
Then \cref{thm:gap-m2} allows us to prove the converse, therefore showing the equivalence stated in
\cref{thm:main-intro}.

\section{\Cref{cond:c} implies spectral gap}\label{sec:cond-c-gap}
\subsection{Quasi-factorization of excitations}
We will start with some useful inequalities regarding orthogonal projectors in Hilbert spaces.

\begin{lemma}\label{lemma:proj}
Let $P$ and $Q$ be two orthogonal projections on a Hilbert space $\hs$.
Then it holds that
\begin{equation}\label{eq:proj}
-\acomm{P}{Q} \le 1-P-Q \le \acomm{P^\perp}{Q^\perp}
\end{equation}
where $\acomm{P}{Q} = PQ +QP $ is the anti-commutator.
\begin{proof}
  We start by observing that $ -1 \le P - Q \le 1$, since $P$ and $Q$ are positive and bounded by
  $1$, and therefore $ 0 \le (P-Q)^2 \le 1$.
  By observing that $(P-Q)^2 = P + Q -\acomm{P}{Q}$, it immediately follows the l.h.s. of
  \cref{eq:proj}:
  \begin{equation}\label{eq:proj-1}
    1-P-Q \ge - \acomm{P}{Q}.
  \end{equation}
  By algebraic manipulation we can show that
  \begin{multline*}
    \{P,Q\} = (1-P^\perp)(1-Q^\perp) + (1-Q^\perp)(1-P^\perp) = \\
    = 2(1 - P^\perp - Q^\perp) + \acomm{P^\perp}{Q^\perp}
    = -2(1- P - Q) + \acomm{P^\perp}{Q^\perp}.
  \end{multline*}
  Applying \cref{eq:proj-1} we obtain that
  \[ \acomm{P^\perp}{Q^\perp} = \acomm{P}{Q} + 2(1-P-Q) \ge 1-P-Q. \]
\end{proof}
\end{lemma}
We are now ready to prove the following quasi-factorization result.

\begin{lemma}[Quasi-factorization of excitations]
\label{lemma:quasi-factorization}
Let $A, B$ be subsets of $\lat$. Then it holds that
\begin{equation}\label{eq:quasi-factorization}
c \expval{P_{A\cup B}^\perp}{\phi} \le
\expval{P_{A}^\perp}{\phi} + \expval{P_{B}^\perp}{\phi},
\end{equation}
where $c = 1-2\delta(A,B).$
\begin{proof}
Notice that frustration freedom implies that
$P_{A\cup B}^\perp P_A^\perp = P_A^\perp P_{A\cup B}^\perp =
P_A^{\perp}$, and the same holds for $P_B^\perp$. Therefore if
$P_{A\cup B}^\perp \ket{\phi} = 0$, both sides of the equation reduce to $0$,
and we can restrict ourselves to the case in which $\ket{\phi}$ is
contained in the image of $P_{A\cup B}^\perp$.
We can then apply \cref{eq:proj} to $P_A^\perp$ and $P_B^\perp$ and
we obtain:
\[
\expval{P_{A\cup B}^\perp}{\phi} \le
\expval{P_A^\perp}{\phi} + \expval{P_B^\perp}{\phi} +
\expval{P_{A\cup B}^\perp \acomm{P_A}{P_B}P_{A\cup B}^\perp}{\phi}.
\]
To conclude the proof, we just need to observe that
\[ P_{A\cup B}^\perp P_A P_B P_{A\cup B}^\perp =
  (P_A - P_{A\cup B})(P_B - P_{A \cup B}),
\]
and that therefore by applying the Cauchy-Schwartz inequality
\begin{multline*}
\expval{P_{A\cup B}^\perp P_A P_B P_{A\cup B}^\perp}{\phi} \\ \le
\norm{P_{A\cup B}^\perp \ket\phi} \norm{(P_A - P_{A\cup B})(P_B - P_{A
    \cup B})P_{A\cup B}^\perp \ket\phi} \\ \le
\norm{(P_A - P_{A\cup B})(P_B - P_{A \cup B})}\norm{P_{A\cup B}^\perp
  \ket\phi}^2 \\
= \norm{(P_A - P_{A\cup B})(P_B - P_{A \cup B})}\expval{P_{A\cup
    B}^\perp}{\phi}.
\end{multline*}
Since the same holds for $P_{A\cup B}^\perp P_A P_B P_{A\cup
  B}^\perp$, and the operator norm is invariant under taking the adjoint, we obtain that
\[ \expval{P_{A\cup B}^\perp \acomm{P_A}{P_B}P_{A\cup
      B}^\perp}{\phi}  \le
2 \norm{(P_A - P_{A\cup B})(P_B - P_{A \cup B})}\expval{P_{A\cup
    B}^\perp}{\phi},  \]
which concludes the proof.
\end{proof}
\end{lemma}

\begin{remark}[Comparison with the Converse DL]
A bound similar to what we have obtained in the previous lemma could
also have been derived from the converse of the
Detectability Lemma (\cref{lemma:converse-dl}). Indeed, if we apply it
to the Hamiltonian $P^\perp_A + P^\perp_B$, we obtain the following
\[ \norm{\phi}^2 - \norm{P_AP_B \ket\phi}^2 \le  4 \expval{P_A^\perp + P_B^\perp}{\phi}. \]
If we now choose $\ket \phi = P_{A\cup B}^\perp \ket\phi$, then a
simple calculation shows that
\begin{multline*}
  \norm{P_AP_B\ket\phi}^2 = \expval{P_BP_AP_B}{\phi} =
  \expval{P_{A\cup B}^\perp P_B P_A P_B P_{A \cup B}^\perp}{\phi} = \\
  \expval{P_B (P_A - P_{A\cup B})(P_B - P_{A \cup B}) P_B}{\phi}
  \le \norm{(P_A-P_{A\cup B})(P_B-P_{A \cup B})} \norm{\phi}^2.
\end{multline*}
We thus obtain the following bound
\begin{equation}
  \label{eq:quasi-fact-detect}
  c^\prime \expval{P_{A\cup B}^\perp}{\phi} \le \expval{P_A^\perp}{\phi} + \expval{P_B^\perp}{\phi},
\end{equation}
but now $c^\prime =  \frac{1}{4}(1-\delta(A,B))$. While very similar
to \cref{eq:quasi-factorization}, the constant $c^\prime$ does not tend to 1 when $\delta(A,B)$ goes
to zero: as we will see next, this is a crucial property and it is for this reason that \cref{eq:quasi-fact-detect}
will not be useful for our proof.
\end{remark}

\begin{remark}
For one dimensional systems, we expect the martingale condition to be implied by exponential decay of correlations, as has been shown in the commuting Gibbs sampler setting \cite{Kastoryano2016}. However, at this point we only know how to obtain this result if for any state $\ket{\psi}$, there exists a (non- Hermitian) operator $f_{A^c}$ on the complement of $A\subseteq\Lambda$ such that
\begin{equation}
P_A\ket{\psi}=f_{A^c}\ket{\varphi},\label{eqn:1Dcd}
\end{equation} and $\ket{\varphi}$ is the unique ground state of $H_\Lambda$. In that case, the proof is analogous to the one in Ref. \cite{Kastoryano2016}. Eqn. (\ref{eqn:1Dcd}) does not hold in general, however it can be shown to hold for injective PEPS. Hence, for injective MPS correlation decay implies the martingale condition.
\end{remark}

\subsection{Spectral gap via recursion}
As we have mentioned in the introduction, the strategy for proving a lower bound to the spectral gap will be a recursive one:
given $\lat$, we will decompose it into two overlapping subsets, so that $\lat = A \cup B$ and we will be able to use \cref{lemma:quasi-factorization}. This would lead to the following expression
\begin{multline}\label{eq:wrong-recursion-1}
  (1-2\delta(A,B)) \expval{P_{\lat}^\perp}{\phi} \le \expval{P_A^\perp}{\phi}  + \expval{P_B^\perp}{\phi} \le \\
  \le \frac{1}{\min(\lambda_A, \lambda_B)} \expval{H_A + H_B}{\phi} =
   \frac{1}{\min(\lambda_A, \lambda_B)} \expval{H_{\lat} + H_{A\cap B}}{\phi}.
\end{multline}
We now face the problem of what to do with the term $\expval{H_{A\cap B}}{\phi}$.
Because of frustration freedom, we can bound it with $\expval{H_{\lat}}{\phi}$, leading to
\begin{equation}
  \label{eq:wrong-recursion-2}
  \lambda_{A\cup B} \ge \frac{1-2\delta(A,B)}{2} \min(\lambda_A, \lambda_B).
\end{equation}
Then it is clear that, even in the case of $\delta(A,B)=0$, this strategy is going to fail:
at each step of the recursion our bound on $\lambda_\lat$ is cut in half,
so in the limit of $\lat \to \inflat$ we will obtain a vanishing lower bound.
The way out of this obstacle is to observe that if we have $s_k$ different ways of splitting
$\lat$ as $A_i \cup B_i$, and if moreover the intersections $A_i \cap B_i$ are disjoint for different $i$,
then we can average \cref{eq:wrong-recursion-1} and obtain
\begin{multline}
\label{eq:right-recursion-1}
  \expval{P_{\lat}^\perp}{\phi} \le \frac{1}{s_k} \sum_{i=1}^{s_k} \frac{(1-2\delta(A_i,B_i))^{-1}}{\min(\lambda_{A_i}, \lambda_{B_i})} \expval{H_{A_i} + H_{B_i}}{\phi} \le \\
  \le \frac{(1-2\delta_k)^{-1}}{\min\{\lambda_{A_i}, \lambda_{B_i}\}_i} \expval{H_\lat +
    \frac{1}{s_k} \sum_{i=1}^{s_k} H_{A_i \cap B_i}}{\phi} \le \\ \le
  (1-2\delta_k)^{-1} \frac{1+1/s_k}{\min\{\lambda_{A_i}, \lambda_{B_i}\}_i}\expval{H_\lat}{\phi}.
\end{multline}
Then \cref{eq:wrong-recursion-2} becomes
\begin{equation}
\label{eq:right-recursion-2}
\lambda_\lat \ge \frac{1-2\delta_k}{1+1/s_k} \min\{\lambda_{A_i}, \lambda_{B_i}\}_i.
\end{equation}
Now the problem is simply to check whether we can find a right balance between the number $s_k$ of different ways to partition $\lat$
(in order to make the product $(1+1/s_k)$ convergent in the recursion), the size of $A_i$ and $B_i$ (if one of them is similar in
size to $\lat$, then we will not have gained much from the recursion), and the size of their overlaps (in order to make $\delta_k$ small).
The geometrical construction presented in \cref{prop:geometrical-construction} shows that such balance is obtainable, if we choose $1/s_k$ to be summable.

By formalizing this idea, we can finally prove the main theorem of this section.
\begin{theorem}[Spectral gap recursion bound]\label{thm:cond-c-gap}
Fix an increasing sequence of positive integers $(s_k)_k$ such that $\sum_k \frac{1}{s_k}$ is summable.
Let $l_k$ and $\mcl F_k$ be as in \cref{def:fk}, and $\delta_k = \delta^s_k$ as in \cref{eq:cond-c} and
\[ \lambda_k = \inf_{\lat \in \mcl F_k} \lambda_\lat .\]
Let $k_0$ be the smallest $k$ such that $\delta_k < 1/2$ for all $k\ge k_0$.
Then there exists a constant $C>0$, depending on $\inflat$ and on the sequence $(s_k)_k$ but not on $k$, such that
\begin{equation}\label{thm:gap-recursion}
\lambda_{k} \ge \lambda_{k_0} C \prod_{j=k_0+1}^k \qty(1 -2\delta_j).
\end{equation}
In particular, if \cref{cond:c} is satisfied, the Hamiltonian is gapped.
\end{theorem}
\begin{proof}
Fix a $\lat \in \mcl F_k\setminus \mcl F_{k-1}$ and let $(A_i, B_i)_{i=1}^{s_k}$ be an $s_k$-decomposition of $\lat$ as in  \cref{prop:geometrical-construction}.
We can then apply \cref{lemma:quasi-factorization} to each pair $(A_i,B_i)$,
average over the resulting bounds, and obtain as in \cref{eq:right-recursion-2}
\[ \lambda_\lat \ge \frac{1-2\delta_k}{1+1/s_k} \min_i \{\lambda_{A_i}, \lambda_{B_i}\} \ge
  \frac{1-2\delta_k}{1+1/s_k} \lambda_{k-1}.  \]
Since $\lat$ was arbitrary, we have obtained that
\begin{equation}\label{eq:proof-recursion-gap}
\lambda_k \ge  \frac{1-2\delta_k}{1+1/s_k} \lambda_{k-1}.
\end{equation}
By iterating \cref{eq:proof-recursion-gap} $k-k_0$ times, we obtain
\[ \lambda_k \ge  \prod_{j=k_0+1}^k \frac{1-2\delta_j}{1+1/s_j} \lambda_{k_0} .\]
We want to show now that this gives rise to the claimed expression. Notice that if we denote $C^{-1} := \prod_{j=1}^\infty (1+1/s_k)$ then
\[  1 \le C^{-1} \le \prod_{j=1}^\infty \qty[ 1+\frac{1}{s_k}] < \infty .\]
This can be seen by observing that the series  $\log \prod_{j=1}^\infty \qty( 1+\frac{1}{s_k}) =  \sum_{j=1}^\infty \log(1+\frac{1}{s_k})$
is summable, since by comparison it has the same behavior as $\sum_j\frac{1}{s_k}$, which is summable by assumption. This implies
in particular that $\prod_{j=1}^k (1+1/s_k)^{-1} \ge C > 0$ for all $k$.
Finally, in order to prove that the Hamiltonian is gapped, we only need to show that \cref{cond:c} implies
\begin{equation}
 K := \prod_{j=k_0+1}^\infty \qty(1 -2\delta_j)  > 0.
\end{equation}
This again is equivalent to the fact that $(\delta_j)_{j=k_0+1}^\infty$ is a summable sequence, which is imposed by \cref{cond:c}.
\end{proof}

\section{Local gap threshold}
Equivalence between \cref{cond:c} and \cref{cond:a} can be seen as a ``self-improving''
condition on $\delta_k$, where assuming that it decays faster than some threshold rate implies
that it is actually decaying exponentially. This type of argument is reminiscent of ``spectral gap amplification'' as described in Ref.\cite{Anshu2016}. The same type of self-improving statement can be
obtained for the spectral gap of $H$.

\begin{lemma}
  \label{cor:knabe-like}
  Fix an increasing sequence of integers $s_k$ such that $\sum \frac{1}{s_k} <\infty$ and $ \sum_{k}\frac{s_k}{l_k}<\infty$.
  Let $H$ be a local Hamiltonian, and let (as in \cref{thm:cond-c-gap}) $\lambda_k = \inf_{\Lambda
    \in \mcl F_k} \lambda_\Lambda$, where $\lambda_\Lambda$ is the spectral gap of $H_\Lambda$.
  If there exist a $C>0$ and a $k_0$ such that
  \begin{equation}\label{eq:cor-knabe-like}
    \lambda_k > C\, k \frac{s_k}{l_k}, \quad \forall k\ge k_0,
  \end{equation}
  then system is gapped (and $\inf_k \lambda_k>0$).
\end{lemma}
\begin{proof}
  Since for every $s_k$-decomposition $A_i, B_i$ of $\Lambda \in \mcl F_k$ the overlap $A_i \cap
  B_i$ has size at least $\frac{l_k}{8 s_k}$, by
  \cref{thm:gap-m2}, we have that
  \[ \delta_k \le \qty(1+\frac{\lambda_k}{g^2})^{-\frac{l_k}{16s_k}} .\]
  We now need to check that $\delta_k$ is summable. By the root test, it is sufficient that
  \[ \limsup \qty(1+\frac{\lambda_k}{g^2})^{-\frac{l_k}{16 k s_k}} = \qty( \liminf \qty(1+\frac{\lambda_k}{g^2})^{\frac{l_k}{26k s_k}} )^{-1} < 1 , \]
  i.e. that
  \[ \liminf \qty(1+\frac{\lambda_k}{g^2})^{\frac{l_k}{16 k s_k}} = \exp(  \liminf \frac{\lambda_k l_k}{16k s_k} \log(1+\frac{\lambda_k}{g^2})^{\frac{1}{\lambda_k}} ) > 1. \]
  If $\liminf \lambda_k = 0$, then $\liminf (1+\frac{\lambda_k}{g^2})^{\frac{1}{\lambda_k}} =
  e^{\frac{1}{g^2}} > 1$ (and if $\liminf \lambda_k >0 $ there is nothing left to prove, since then we already know that the system is gapped ), and therefore we can reduce to check that
  \[ \liminf \frac{\lambda_k}{k} \frac{l_k}{s_k} > 0, \]
  which is implied by \cref{eq:cor-knabe-like}.
\end{proof}

If we now read the condition of \cref{eq:cor-knabe-like} in terms of the length of the sides of the sets in $\mcl F_k$, we obtain a proof of \cref{cor:threshold}.
\begin{proof}[Proof of \Cref{cor:threshold}.]
  Let $\Lambda \in \mcl F_k$: then its diameter will be at most a
  constant times $l_{k}$. If we denote it by $n$, then
  $k \ge q \log(n)$ for some $q>0$. If we choose
$s_k = k^{1+\epsilon}$ for some $\epsilon >0$, we see that
\cref{eq:cor-knabe-like} is satisfied if we can find $\epsilon$ and
$C>0$ such that
\begin{equation}\label{eq:knabe-bound}
  \lambda_\Lambda > C \frac{\log(n)^{2+\epsilon}}{n}
\end{equation}
holds for all rectangles $\Lambda$ with sides bounded by
$n$. If the Hamiltonian is gapless, then necessarily
$\lambda_\Lambda = o\qty(\frac{\log(n)^{2+\epsilon}}{n})$ for every
$\epsilon > 0$.
\end{proof}

This result has to be compared with similar results obtained in Refs. \cite{Knabe1988,Gosset2016} in the specific
case of nearest-neighbor interactions in 1D chains and in 2D square and hexagonal lattices. In all
these cases, the authors obtain a local gap threshold which implies a spectral gap in the limit in the
following sense: denoting $\lambda_n$ the spectral gap of a finite system defined on a subset of
``side-length'' $n$ (where the exact definition depends on the dimension and the geometry of the
lattice, but the general idea is that such a subset has $\order{n^D}$ sites), there exists a sequence
$\gamma_n$ (the local gap threshold) such that, if for some $n_0$ it holds that
$\lambda_{n_0} > \gamma_{n_0}$, then the system is gapped in the limit. The converse is that, if the
Hamiltonian is gapless, then $\lambda_n = \order{\gamma_n}$. The values of $\gamma_n$ present in
Refs. \cite{Knabe1988,Gosset2016} are recalled in \cref{table:knabe}.

The obvious downside of \cref{cor:knabe-like} over the results in Refs. \cite{Knabe1988,Gosset2016}
is that these only require a single $n_0$ satisfying $\lambda_{n_0} > \gamma_{n_0}$, while
\cref{eq:cor-knabe-like} is a condition to be satisfied for each $n$. On the other hand, it can be
applied in more general settings than nearest neighbor interactions, as well as in dimensions higher
than 2, and can be easily generalized to regions with different shapes. The upper bound on
$\lambda_n$ for a gapless Hamiltonian which we derive is worse by a polynomial factor than the ones
obtained in 1D and in the 2D square lattice, and it is only off by a logarithmic factor than the 2D
hexagonal lattice case.  While the logarithmic factor in our bound is probably just an artifact of
the proof, it is an interesting open question whether the optimal scaling for the general case is $\order{1/n^2}$.

\begin{table}[h]
  \begin{center}
  \begin{tabular}{*4c}
    \hline
     1D \cite{Knabe1988} & 1D \cite{Gosset2016} & 2D hexagonal \cite{Knabe1988} & 2D square \cite{Gosset2016}   \\ \hline
     $\frac{1}{n-1}$ & $\frac{6}{n(n+1)}$ & $\frac{1}{3n-1}$ & $\frac{8}{n^2}$
  \end{tabular}
  \caption{Local gap thresholds for different spin lattices (values for $\gamma_n$).}
  \label{table:knabe}
\end{center}
\end{table}

One should also mention the Lieb-Schultz-Mattis theorem \cite{Lieb_1961} and its generalization to
higher dimensions \cite{Hastings_2004,Nachtergaele_2007}, which proves that a class of half-integer
spin models (not necessarily frustration free) with translational invariance, continuous symmetry
and unique ground state is gapless. For this class of models the gap is bounded by
$\order{\frac{\log n}{n}}$ (the $\log n$ factor can be removed in 1D), which is slightly better than
the general upper bound we have obtained.

\section*{Acknowledgments}
  We acknowledge financial support from the European Research Council (ERC Grant Agreement no
  337603), the Danish Council for Independent Research (Sapere Aude) and VILLUM FONDEN via the QMATH
  Centre of Excellence (Grant No. 10059) and the Villum Young Investigator Award.
  A.L. would like to thank Amanda Young for the fruitful discussion about the PVBS model.

\printbibliography

\appendix
\section{Geometrical construction}
\label{appendix:geometrical-construction}
\begin{proof}[Proof of \cref{prop:geometrical-construction}]
Let $d_k = \frac{l_k}{8s}$. For $i=1,\dots,s$, we define
\begin{align*}
  A_i &= \qty([0,l_{k+1}]\times \dots \times [0,l_{k+D-1}] \times \qty[0, \frac{l_{k+D}}2 + 2i\, d_k]) \cap \lat ; \\
  B_i &= \qty([0,l_{k+1}]\times \dots \times [0,l_{k+D-1}] \times \qty[\frac{l_{k+D}}2 + (2i-1)d_k, l_{k+D}]) \cap \lat .
\end{align*}
Let us start by proving that $A_i$ and $B_i$ are in $\mcl F_{k-1}$. In order to do so, we need to
show that up to translations and permutations of the coordinates, they are contained in
$R(k-1)$. If we look at coordinate $j=1,\dots, D-1$, then their sides are
contained in $[0,l_{k+j}]$, so it is enough to show that across the $D$-th coordinate they are not
more than $l_k$ long. $A_i$ has a larger side than $B_i$, so we can focus on it only. Then we see
that
\[ \frac12 l_{k+D} + 2i\, d_k \le \frac{1}{2}(3/2)^\frac{k+D}{D} + 2s\, d_k = \frac34 l_k + \frac14
  l_k = l_k. \] So that $A_i$ and $B_i$ belong to $\mcl F_{k-1}$ for every $i$.  If either $A_i$ or
$B_i$ were empty for a given $i$, then $\lat$ would be contained in a set belonging to
$\mcl F_{k-1}$, and thus it would itself belong to $\mcl F_{k-1}$, but we have excluded this by
assumption. So $A_i$ and $B_i$ are not empty.  Clearly $\lat = A_i \cup B_i$, and
$\dist(\lat\setminus A_i, \lat \setminus B_i) \ge d_k -2$. Finally, we see that
\[ A_i \cap B_i = \qty([0,l_{k+1}]\times \dots \times [0,l_{k+D-1}] \times \qty[\frac{l_{k+D}}2 + (2i-1)d_k,  \frac{l_{k+D}}2 + 2i\, d_k]) \cap \lat,\]
so that $A_i \cap B_i \cap A_j \cap B_j = \emptyset$ for all $i\neq j$.
\end{proof}

\end{document}